\documentclass[journal,twoside,web]{ieeeconf}

\usepackage{cite}
\usepackage{amsmath,amssymb,amsfonts}
\usepackage{algorithm, algpseudocode}
\usepackage{graphicx,booktabs,colortbl}
\usepackage{textcomp}

\DeclareMathOperator*{\argmin}{arg\,min}
\def\BibTeX{{\rm B\kern-.05em{\sc i\kern-.025em b}\kern-.08em
    T\kern-.1667em\lower.7ex\hbox{E}\kern-.125emX}}
\usepackage[english]{babel}
\usepackage{mathrsfs}

\newcommand{\R}{{\mathbb R}}
\newcommand{\e}{{\varepsilon}}
\newcommand{\T}{\mathrm T}
\newcommand{\tail}{\mathcal T}
\newcommand{\head}{\mathcal H}
\newcommand{\spec}{\mathrm{spec}}

\usepackage{dsfont}
\newcommand{\mathbbm}{\mathds}

\newtheorem{propo}{Proposition}
\newtheorem{defi}{Definition}

\newtheorem{coro}{Corollary}
\newtheorem{remark}{Remark}
\newtheorem{theorem}{Theorem}

\begin{document}
\title{Optimal pinning control of directed hypergraphs}
\author{Fabio Della Rossa,
Davide Liuzza,
Francesco Lo Iudice, and
Pietro De Lellis, \IEEEmembership{Senior Member, IEEE}
\thanks{This study was carried out within the 2022FHHHPC «The Structure, Dynamics and Control of Network Systems With Higher-Order Interactions» and the 2022K8EZBW «Higher-order interactions in social dynamics with application to monetary networks» projects – funded by European Union – Next Generation EU  within the PRIN 2022 program (D.D. 104 - 02/02/2022 Ministero dell’Università e della Ricerca). This manuscript reflects only the authors’ views and opinions and the Ministry cannot be considered responsible for them.
The work of Davide Liuzza was supported in part by the University of Sannio—“Finanziamento della Ricerca di Ateneo (FRA).}
\thanks{P. De Lellis is with the Department of Electrical Engineering and Information Technology, University of Naples Federico II, 80125, Naples, Italy (e-mail: pietro.delellis@unina.it).}
\thanks{F. Della Rossa is with the Department of Electronics, Information, and Bioengineering, Politecnico of Milan, Italy (e-mail: fabio.dellarossa@polimi.it).}
\thanks{F. Lo Iudice is with the Department of Electrical Engineering and Information Technology, University of Naples Federico II, 80125, Naples, Italy (e-mail: francesco.loiudice2@unina.it).}
\thanks{D. Liuzza is with the Department of Engineering, University of Sannio, Benevento, Italy.}}

\maketitle 

\begin{abstract}
Identifying the nodes that must be directly controlled to steer a network along a desired trajectory remains an open problem for digraphs, and even more so for hypergraphs. In this manuscript, we investigate network systems coupled via directed hypergraphs and consider a broad class of individual dynamics and coupling configurations, extending the definition of type II networks originally formulated for digraphs. For this class of networks with higher-order interactions, we establish necessary and sufficient conditions under which a pinning selection locally ensures successful control. Building on these analytical results, we propose a greedy heuristic for pinning control selection, which demonstrably outperforms existing methods.
\end{abstract}
\section{Introduction}
Steering the collective behavior of networks of coupled dynamical systems (also called multiagent systems) has been an open challenge in the control literature in the last decades \cite{d2023controlling}. The interest of the research community has been triggered by the numerous applications spanning wide areas of science and engineering. For instance, coordinating the motion of interacting systems is at the foundation of several engineering problems, including formation control \cite{cao2011formation}, rendez-vous problems \cite{canuto2008eulerian}, and power grid synchronization \cite{nishikawa2015comparative}. Additionally, control theory has provided insights into the mechanisms underlying opinion dynamics \cite{parsegov2016novel}, and into strategies for controlling and mitigating epidemic outbreaks \cite{della2020network,parino2021modelling}.

When the control action can only be exerted on a (possibly small) fraction of the network nodes, the classical pinning control scheme has been proposed \cite{wang2002pinning}. In pinning control, an external node, the \emph{pinner}, sets a solution of the individual dynamics as the reference trajectory for the entire network. The pinner then sends a proportional control action to a fraction of the network nodes, denoted \emph{pinned nodes}, of which it is assumed to be able to measure the state. Each feedback signal is modeled as a directed edge from the pinner to a pinned node. 

The problem then became that of identifying how many and which nodes to pin in order to control the entire network. It was observed that at least one node in each root strongly connected component of the graph should be pinned \cite{lu2010global}. Then, given a number of node measurements that can be taken, researchers have tried to optimally select where to inject the control signals. For instance, in the context of consensus, the maximization of the smallest eigenvalue of the grounded Laplacian has been used as a selection criterion for the pinned nodes \cite{pirani2015smallest,zhou2025maximizing}. Adaptive strategies have also been proposed to modulate the gains to foster pinning controllability \cite{di2019decentralized}.

Existing approaches focus on standard pinning control on digraphs. However, it has been recently pointed out that, in several applications, this control problem should be formulated on directed hypergraphs instead of digraphs, for a twofold reason. First, modeling the connections between the controlled nodes by directed edges, implicitly assumes pairwise interactions. In several real-world network systems, however, nonlinear interactions take place between groups (larger than 2) of agents, which cannot be factored as sum of pairwise interactions. This is the case, for instance, of chemical reaction networks, where reaction terms depend on the joint concentration of multiple reactants \cite{feinberg2019foundations}, or in social contagion and opinion formation, where adoption or behavioral change may require reinforcement from multiple peers at once \cite{iacopini2019simplicial,neuhauser2022consensus}.

The second reason for which directed multibody interactions naturally arise in pinning control is that assuming directed edges from the pinner to the controlled nodes would imply the ability of individually measuring the outputs of the pinned nodes. This is not the case, however, when the sensors can only measure an aggregated output of two or more nodes \cite{della2023emergence}. Such relevant case cannot be modeled by directed edges from the pinner to the pinned nodes, but rather as directed hyperedges from the pinner to a set of pinned nodes. This control framework has been recently proposed in \cite{de2022pinning,della2023emergence,liuzza2025synchronization} and references therein, but no method exists for the minimal selection of the pinning hyperedges, that is, a strategy that minimizes the number of measurements required for control. The only existing heuristic, proposed in \cite{de2022pinning}, can only be applied to the case of pinning edges, that is, when the pinner can individually measure the state of the pinned nodes.

In this paper, we aim to fill this gap and tackle the problem of optimally selecting the nodes to constitute the pinning hyperedges. In particular, we focus on network systems with a type II Master Stability Function, a definition that we extend to the case of higher-order interactions. For this wide class of systems (which include e.g. diffusively coupled smooth QUAD dynamical systems), we study the limit behavior of the eigenvalues of an extended Laplacian that describes the pinning controlled hypergraph. Based on this result, we are able to formulate an optimal control problem so to use a minimal number of measurements to control the network to the desired trajectory. Somewhat surprisingly, we find that there are instances in which pinning hyperedges are preferred to pinning edges, albeit the former are associated to aggregated, lower-resolution measurements.

Since the optimal control problem should be in principle solved by exhaustive search, unfeasible for large networks, we also propose an analytically grounded, yet efficient greedy heuristic. Through extensive numerical analysis on testbed hypergraph topologies, we show that a) the proposed heuristic closely matches the performance of an exhaustive search, b) when applied to the selection of pinning hyperedges, it strongly outperforms the heuristic proposed in \cite{de2022pinning}, as well as alternative, purely topological, strategies, and c) it also effectively works for the selection of 
pinning hyperedges, where it closely matches the optimal solution obtained from exhaustive search.
\section{Mathematical preliminaries}
\subsection{Matrix notation}
Given a positive integer $n$, $I_n$ denotes the identity matrix in $\R^{n\times n}$, with $\mathfrak e_i$ being its $i$-th column, and $0_n$ and $\mathbbm{1}_n$ are a vector of all zeros and of all ones in $\R^n$, respectively. Given a vector $v=[v_1,\ldots,v_n]^\T$, $\mathrm{diag}(v)\in\R^{n\times n}$ is the diagonal matrix with elements  $v_1,\ldots v_n$ on the diagonal. 
Given a matrix $M\in\R^{n\times n}$, $M^\T$ is its transpose, and $\sigma(M)$ is the set containing its eigenvalues. Further, we denote $\lambda_i(M)$ the $i$-th eigenvalue of $M$.
%sort the eigenvalues $\lambda_1(M),\ldots,\lambda_n(M)$ of matrix $M$ so that their real part is in ascending order, that is, $\Re(\lambda_1(M))\le\ldots\le\Re(\lambda_n(M))$.

Given two matrices $M_1\in\R^{a\times b}$ and $M_2\in\R^{c\times d}$, we denote $(M_1\otimes M_2)\in\R^{ac\times bd}$ their Kronecker product \cite{horn2012matrix}, and, when they have the same number of columns ($b=d$), $[M_1;M_2]\in\R^{(a+c)\times b}$ their vertical concatenation, whereas, when they have the same number of rows $(a=c)$, $[M_1,M_2]\in\R^{a\times(b+d)}$ their horizontal concatenation.

\subsection{Directed hypergraphs}\label{subsec:directed_hyper}
A directed hypergraph $\mathscr H$ is a pair $(\mathcal V,\mathcal E)$, where $\mathcal V=\{\nu_1,\ldots,\nu_N\}$ is the set of nodes, and $\mathcal{E}=\{\e_1,\ldots,\e_M\}$ is the set of directed hyperedges; the $i$-th directed hyperedge $\e_i$ of $\mathscr H$ is an ordered pair $(\tail(\e_i),\head(\e_i))$ of (possibly empty) disjoint subsets of the hypergraph nodes \cite{gallo1993directed}. Namely, the ordered subsets $\tail(\e_i)$ and $\head(\e_i)$ of $\mathcal{V}$, are the set of tails and heads of the hyperedge $\e_i$, respectively, and such that $\tail(\e_i)\cap \head(\e_i)=\emptyset$. The cardinality $|\e_i|$ of a hyperedge $\e_i$ is defined as the number $|\tail(\e_i)|+|\head(\e_i)|$ of nodes composing it. The functions $\mathfrak t(\e,i)$ and $\mathfrak h(\e,j)$ associate to the $i$-th tail and $j$-th head of a hyperedge $\e\in\mathcal E$ the corresponding labels in $\mathcal V$, respectively. Furthermore, given two node subsets $\mathcal V_1,\mathcal V_2\subseteq \mathcal V$, we denote $\mathcal E^{\mathcal V_1,\mathcal V_2}=\{\e \in\mathcal E:\mathcal V_1\subseteq \mathcal \tail(\e)\land \mathcal V_2\subseteq \head(\e)\}$; with a slight abuse of notation, when a subset is a singleton, we will refer to it by its only element, e.g. if $\mathcal V_1=\{\nu_j\}$ we write $\mathcal E^{j,\mathcal V_2}$ in place of $\mathcal E^{\{\nu_j\},\mathcal V_2}$. Finally,
we denote $\mathcal{E}^{\cdot,j}=\{\e_i\in\mathcal E : \nu_j\in \head(\e_i) \}$ as the subset of hyperedges having $\nu_j$ as a head, and $\mathcal{E}^{j,\cdot}=\{\e_i\in\mathcal E : \nu_j\in \tail(\e_i) \}$ as the subset of hyperedges having $\nu_j$ as a tail; we define the in-degree $d_j^{\mathrm{in}}$ and out-degree $d_j^{\mathrm{out}}$ of a node $\nu_j$ as the cardinality of the latter two sets, that is, $d_j^{\mathrm{in}}=|\mathcal{E}^{j,\cdot}|$ and $d_j^{\mathrm{out}}=|\mathcal{E}^{\cdot,j}|$. 
%Similarly, $\mathcal{E}^{i,j}=\{\e_k\in\mathcal E : \nu_i\in \tail(\e_k)\land \nu_j\in \head(\e_k) \}$ is the subset of hyperedges having $\nu_i$ as a tail and $\nu_j$ as a head. 
%\textcolor{blue}{Why are we not interested in defining $\mathcal{E}^{j,\cdot}$?}\red{It seems that we do not need it, since what we often need is the opposite, which is a sort of equivalent of the in-neighborhood for hypergraphs}
\subsection{Signed graphs \cite{ahmadizadeh2017eigenvalues}}
A weighted signed graph $\mathscr S$ is defined by the triple $\{\mathcal V,\mathcal E, \mathcal W\}$, where $\mathcal V$ is the set of nodes, $\mathcal E\subseteq \mathcal V\times \mathcal V$ is the set of edges, and the function $\mathcal W: \mathcal V\times \mathcal V \rightarrow \mathbb R$ associates $0$ to each pair $(i,j)\in\mathcal V\times \mathcal V$ that is not in $\mathcal E$, and a non-zero weight to each edge in $\mathcal E$. Different from standard weighted digraphs, also negative weights can be associated to edges. The adjacency matrix $A$ associated to $\mathscr S$ is such that its $ij$-th entry $a_{ij}$ is equal to the weight $\mathcal W(i,j)$ associated to edge $(i,j)$.
The Laplacian matrix for signed graphs has been defined as $L=D-A$, where $D=\mathrm{diag}\left(\left[d_1^{\mathrm{out}},\ldots,d_{|\mathcal V|}^{\mathrm{out}}\right]^\T\right)$, with $d_i^{\mathrm{out}}=\sum_{j=1}^{|\mathcal V|}a_{ij}$ being the out-degree of node $i$. By definition, $L$ is zero row-sum, which implies that $0\in\spec(L)$, and that $\mathbbm{1}_N$ is its associated (right) eigenvector.

% The superimposition of two signed graphs $\mathscr S_1=\{\mathcal  V,\mathcal E_1,\mathcal W_1\}$ and $\mathscr S_2=\{\mathcal  V,\mathcal E_2,\mathcal W_2\}$ is a graph $\mathscr S=\mathscr S_1 \oplus\mathscr S_2=\{\mathcal V, \mathcal E, \mathcal W\}$ where $\mathcal E=\mathcal E_1\cup\mathcal E_2$, and $\mathcal W(i,j)=\mathcal W_1(i,j)+\mathcal W_2(i,j)$. 

\section{Pinning controllability of network systems on hypergraphs}
\subsection{Network model and control question}\label{subsec:model_question}
We consider an ensemble of $N$ nonlinear dynamical systems coupled through a directed hypergraph $\mathscr H_c=\{\mathcal V_c,\mathcal E_c\}$, where the sets $\mathcal V_c=\{\nu_1,\ldots,\nu_N\}$ and $\mathcal{E}_c=\{\e_1,\ldots,\e_N\}$ are the set of nodes and hyperedges, respectively. Given a node $\nu_i\in\mathcal V_c$, its state will be a vector $x_i \in \mathbb{R}^n$, whereas, denoting $x_{\mathfrak{t}(\e,i)}$ and $x_{\mathfrak h(\e,j)}$ the state of the $i$-th tail and of the $j$-th head of a hyperedge $\e\in\mathcal E_c$,
we associate to $\e$ a tail state matrix $x_\e^\tau=[x_{\mathfrak{t}(\e,1)},\ldots,x_{\mathfrak{t}(\e,|\tail(\e)|)}]$ and a head state matrix $x_\e^h=[x_{\mathfrak{h}(\e,1)},\ldots,x_{\mathfrak{h}(\e,|\head(\e)|)}]$. In the absence of a control action, the network dynamics would read
\begin{equation}\label{eq:uncontrolled_network}
    \dot{x}_i= f(x_i) + \sum_{\e\in\mathcal E_c^{\cdot,i}}\sigma_{\e}g(x_{\e}^{\tau}\alpha_\e- x_{\e}^h\beta_\e),
\end{equation}
where $\sigma_{\e}$ is the coupling gain associated to the hyperedge $\e$; $f:\mathbb R^n\times \mathbb R_{\ge 0}\rightarrow \mathbb R^n$ is the vector field describing the individual dynamics, and $g:\mathbb R^n\rightarrow \mathbb R^n$ is the (possibly nonlinear) coupling protocol;
$\alpha_\e=[(\alpha_\e)_{\mathfrak t(\e,1)},\ldots,(\alpha_\e)_{\mathfrak t(\e,|\tail(\e)|)}]^\T$ and $\beta_\e=[(\beta_\e)_{\mathfrak h(\e,1)},\ldots, (\beta_\e)_{\mathfrak h(\e,|\head(\e)|)}]^\T$ are the (ordered) vectors stacking the non-negative weights associated to the tails and heads of $\e$, respectively, defined such that $\alpha_\e^\T \mathbbm{1}_{|\tail(\e)|}=\beta_\e^\T \mathbbm{1}_{|\head(\e)|}=1$. This coupling protocol is known as hyperdiffusive, as it reduces to the standard diffusive protocol when only pairwise interactions are present \cite{della2023emergence}.

\begin{defi}\label{rem:homo}
The hypergraph weights are \textit{homogeneous} when $\alpha_\e=\mathbbm{1}_{|\tail(\e)|}/|\tail(\e)|$ and $\beta_\e=\mathbbm{1}_{|\head(\e)|}/|\head(\e)|$ \cite{de2022pinning}.  
\end{defi}

In the presence of a pinner, an additional node $p$ is added to the network, with state $x_p\in\mathbb{R}^n$, sharing the same individual dynamics as the rest of the network, and setting the reference trajectory as the solution of the following Cauchy problem
\begin{equation}
\dot x_p = f(x_p,t),\qquad x_p(0)=x_{p0}.
\end{equation}
The pinner node is unidirectionally coupled to a subset of the network nodes, $\mathcal P$, through a set of directed hyperedges (the \textit{pinning hyperedges} $\mathcal E_{\mathrm{pin}}=\{\varepsilon_1^p,\ldots,\varepsilon_m^p\}$), characterized by the fact that the pinner is the sole tail. The heads sets of the pinning hyperedges are obtained from a partition of the set $\mathcal P$, this meaning that the same node cannot be pinned by two different hyperedges, and that $|\mathcal H(\varepsilon_1^p)|+\ldots+|\mathcal H(\varepsilon_m^p)|=|\mathcal P|$.

When the pinner is present, the network dynamics \eqref{eq:uncontrolled_network} then become
\begin{equation}\label{eq:controlled_network}
\begin{aligned}
    \dot{x}_i&= f(x_i) + \sum_{\e\in\mathcal E_c^{\cdot,i}}\sigma_{\e}g(x_{\e}^{\tau}\alpha_\e- x_{\e}^h\beta_\e)\\
    &+\sum_{\varepsilon\in\mathcal E_{\mathrm{pin}}}\kappa g(x_p-x_{\varepsilon}^h\beta_{\varepsilon}),
\end{aligned}
\end{equation}
where $\kappa>0$ is the control gain.

Since the hyperdiffusive coupling protocol is synchronization noninvasive, the synchronization manifold $x_i(t)=x_p(t)$ for all $i\in\mathcal V$ is invariant. The control goal is to drive the pinning error $e(t)=[e_1;\ldots;e_N]$, with $e_i=x_i-x_p$, to zero. More formally,
\begin{defi}
    The controlled network \eqref{eq:controlled_network} is locally asymptotically controlled to the pinner's trajectory if
    \begin{enumerate}
        \item for each $\varsigma>0$ and $t_0\ge 0$, there exists $\Delta_1(\varsigma,t_0)$ such that
        \begin{equation}
            \|e(t)\|<\varsigma,\quad \forall t\ge t_0,
        \end{equation}
        for $\|e(t_0)\|\le \Delta_1(\varsigma,t_0)$; and
        \item for each $t_0\ge 0$, there exists a $\Delta_2(t_0)$ such that
        \begin{equation}
            \lim_{t\rightarrow +\infty}e(t)=0.
        \end{equation}
        for $\|e(t_0)\|\le \Delta_2(t_0)$.
    \end{enumerate}
\end{defi}

\textit{Control question:} is network  \eqref{eq:controlled_network} pinning controllable, in the sense that there exists a control gain $\kappa$ for which \eqref{eq:controlled_network} is locally asymptotically controlled?
%, with $s$ being any solution of the decoupled dynamics $\dot s=f(s,t)$, it makes sense to linearize \eqref{eq:controlled_network} the dynamics along $s$. 
\subsection{Convergence analysis}
Notice that the argument of the nonlinear function $g$ in the first summation of \eqref{eq:controlled_network} can be rewritten as
\begin{equation}\label{eq:equiv}
\sum_{j\in\mathcal T(\e)}(\tilde\alpha_{\e})_j(x_j-x_i)-\sum_{j\in\mathcal H(\e)}(\tilde\beta_{\e})_j(x_j-x_i),
\end{equation}
where $\tilde \beta_\e$ ($\tilde \alpha_\e$) belongs to $\R^N$ and its $j$th element is 0 if node $j$ is not a tail (head) of $\e$, whereas, if $j$ is a tail (head) of $\e$, it is equal to the weight associated to that tail (head).
Moreover, we assume without loss of generality that the first $|\mathcal P|$ nodes receive a control input, sorted so that the nodes in $\varepsilon_i^p$ precede those in $\varepsilon_{i+1}^p$, for all $i=1,\ldots,m-1$. It is then possible to define the following pinning matrix
% \begin{equation}\label{eq:pinning_matrix}
%     P=\begin{bmatrix}
%         P_{\varepsilon_1^p} & \cdots & \cdots & & & \\
%         \vdots & \ddots & \vdots & & 0_{|\mathcal P|\times(N-|\mathcal P|)} &\\
%         \cdots & \cdots & P_{\varepsilon_m^p} & & &\\
%         & & & & & \\
%         & 0_{(N-|\mathcal P|)\times |\mathcal P|} & & & 0_{(N-|\mathcal P|)\times (N-|\mathcal P|)} &\\
%         & & & & &
%     \end{bmatrix},
% \end{equation}
\begin{equation}\label{eq:pinning_matrix}
P=\begin{bmatrix}
    P_{\varepsilon} & 0_{|\mathcal P|\times(N-|\mathcal P|)}\\
    0_{(N-|\mathcal P|)\times |\mathcal P|} & 0_{(N-|\mathcal P|)\times (N-|\mathcal P|)}
\end{bmatrix},
\end{equation}
where $ P_\varepsilon\in\mathbb R^{|\mathcal P|\times|\mathcal P|}$ is the block diagonal matrix
\begin{equation*}
    P_\varepsilon=\begin{bmatrix}
        P_{\varepsilon_1^p} & &\\
        & \ddots & \\
        & & P_{\varepsilon_m^p}
    \end{bmatrix}\!,\text{ with }
%\end{equation}
%\begin{equation}
    P_{\varepsilon_i^p}=\begin{bmatrix}
        \beta_{\varepsilon_i^p}^{(1)} & \cdots & \beta_{\varepsilon_i^p}^{(|\varepsilon_i^p|)}\\
        \vdots &  & \vdots\\
        \beta_{\varepsilon_i^p}^{(1)} & \cdots & \beta_{\varepsilon_i^p}^{(|\varepsilon_i^p|)}
    \end{bmatrix}\!. 
\end{equation*}
We can now focus on the dynamics of the $i$-th pinning error $e_i=x_i - x_p$ and linearize its dynamics around the reference trajectory $x_p$, thus obtaining
\begin{equation}\label{eq:deltax_dyn}
\dot {e}_i =
 \mathrm{JF}(x_p) e_i-\sum_{j=1}^N (L_{ij}+\kappa P_{ij}) \mathrm{JG}(0)  e_{j},
 \end{equation}
where $\mathrm{JF}\in\R^{n\times n}$ and $\mathrm{JG}\in\R^{n\times n}$ are the Jacobian matrices associated to $f$ and $g$, respectively, and $L_{ij}$ is the entry $ij$ of the Laplacian matrix of the signed graph $\mathscr S(\mathscr H_c)$ associated to $\mathscr H_c$, defined as
\begin{equation}\label{eq:lapl_signed}
L_{ij}=\sum_{\e\in\mathcal E^{\cdot,\{i,j\}}}(\tilde\beta_\e)_j \sigma_\e-\sum_{\e\in\mathcal E^{j,i}}(\tilde\alpha_\e)_j \sigma_\e,
\end{equation}
with $\mathcal E^{j,i}$ being the set of hyperedges having $j$ as a tail and $i$ as a head, $\mathcal E^{\cdot,\{i,j\}}$ is set of hyperedges having both $i$ and $j$ as heads.

We can then rewrite \eqref{eq:deltax_dyn} in matrix form as 
\begin{equation}\label{eq:contr_mat_form}
\dot{e} = \big(I_N\otimes \mathrm{JF}(x_p)\big) e - M\otimes \mathrm{JG}(0)e,    
\end{equation}
where $M(\kappa)=L+\kappa P$.

Next, we introduce the matrix $V$ such that $V^{-1}M V$ is in Jordan form, and then introduce the transformation $\eta=(V\otimes I_n)e$. In general, matrix $M$ has $\nu\le n$ Jordan blocks, and therefore $\eta$ can be decomposed as $[\eta_1,\ldots,\eta_\nu]$, where $\eta_i=[\eta_{i1},\ldots,\eta_{i c_i}]\in\mathbb C^{nc_i}$ are the components of $\eta$ associated to the $i$-th block, with $\eta_{ij}\in\mathbb C^n$ and $c_i$ being the size of the $i$-th block. The dynamics of $\eta_i$ are then given by
\begin{align} \label{eq_etas}
\dot{\eta}_{i1} & =\left(\mathrm{JF}(x_p)-\lambda_i \mathrm{JG}(0)\right) \eta_{i1},\nonumber \\
\dot{\eta}_{i2} & =\left(\mathrm{JF}(x_p)-\lambda_i \mathrm{JG}(0)\right) \eta_{i2}-\mathrm{JH}(0) \eta_{i1}, \nonumber\\
\vdots & \\
\dot{\eta}_{i c_i} & =\left(\mathrm{JF}(x_p)-\lambda_i \mathrm{JH}(0)\right) \eta_{i c_i}-\mathrm{JG}(0) \eta_{i (c_i-1)},\nonumber
\end{align}
where $\lambda_i$ is the $i$-th eigenvalue of $M$.

We can now introduce the following master equation:
\begin{equation}\label{eq:msf}
    \dot \xi = \big(\mathrm{JF}(x_p) - \mu \mathrm{JG}(0)\big)\xi,
\end{equation}
where $\xi\in\mathbb R^n$ and $\mu\in\mathbb C^n$. Following the notation used e.g. in \cite{pecora1998master} for networks on digraphs and in \cite{della2023emergence} for networks on directed hypergraphs, we call the maximum Lyapunov exponent $\Lambda(\mu)$ associated with \eqref{eq:msf} the master stability function for the pinned network \eqref{eq:controlled_network}.

We can now give the following proposition:
\begin{propo}\label{prop:1}
    if $\Lambda_{\max}:=\max_{\lambda\in\sigma(M(\kappa))} \Lambda(\lambda)<0$, then the controlled network \eqref{eq:controlled_network} is locally asymptotically controlled to the pinner's trajectory.
\end{propo}
\begin{proof}
    As $\Lambda_{\max}<0$, all the Jordan blocks are asymptotically stable, and the thesis follows.
\end{proof}

To answer the control question stated at the end of Section \ref{subsec:model_question} using Proposition \ref{prop:1}, one would need to check its hypothesis for all possible values of $\kappa$. To avoid such a daunting task, we can check the assumption in the limit of an infinitely large control gain. From the continuity of the eigenvalues of $M(\kappa)$, it would then follow the existence of a finite $\kappa$ guaranteeing network pinning controllability, as stated in the following:

\begin{coro}\label{prop:2}
    if $\lim_{\kappa\rightarrow+\infty}\max_{\lambda\in\sigma(M(\kappa))}\Lambda(\lambda)<0$, then network \eqref{eq:controlled_network} is pinning controllable.
\end{coro}
% \begin{proof}
%     From Proposition \ref{prop:1} and as the eigenvalues vary with continuity with respect to $\kappa$, the thesis follows.
% \end{proof}
Corollary \ref{prop:2} could then be used to determine whether a given pinning configuration, described by the set $\mathcal E_{\mathrm{pin}}$, can guarantee local asymptotic control to the desired trajectory. However, in order to use this result, we need to compute the spectrum of $M$ as the control gain $\kappa$ goes to infinity, which is the focus of the next subsection. 
\subsection{Limit spectrum of $M=L+\kappa P$}
Let us study the eigenvalues of $M(\kappa)$ for $\kappa\rightarrow \infty$. From \eqref{eq:pinning_matrix}, and reminding that $\beta_\epsilon^\T \mathbbm{1}_{|\mathcal H(\epsilon)|}=1$ for all $\epsilon \in\mathcal E$, there exists a transformation $T$ that diagonalizes $P$, such that we can read on its diagonal  $m$ unitary eigenvalues (each corresponding to a pinning hyperedge) followed by its $N-m$ zero eigenvalues. We can then apply this transformation to matrix $M$, and obtain $\overline M=T^{-1} M T$, which has the same spectrum as $M$. Matrix
$\overline M$ can be written as
\[
\overline M(\kappa)=\overline L + \kappa \overline P,
\]
where $\overline P=\mathrm{diag}([\mathbbm{1}_m;0_{N-m}])$, and $\overline L=T^{-1} L T$ can be written as
\[
\overline L=\begin{bmatrix}
    \overline L_{11} & \overline L_{12}\\
    \overline L_{21}& \overline L_{22}
\end{bmatrix}.
\]
where $L_{11}$ is a square $m$-dimensional block and the dimensions of the remaining blocks follow. We are now ready to state the following result
\begin{theorem}\label{thm:asymptotic_spectrum}
    Denoting $\lambda_i(M(k))$, $i=1,\ldots,N$, the eigenvalues of $M(\kappa)$, we have that\footnote{In \eqref{eq:lemma_nm}, the limit of a set of functions is intended as the set of the limits of each element of the set.}
    \begin{subequations}
    \begin{align}
        &\lim_{\kappa\rightarrow+\infty}\lambda_i(M(\kappa))=+\infty, \quad i=1,\ldots,m,\label{eq:lemma_m}\\
        &\lim_{k\rightarrow+\infty}\{\lambda_{m+1}(M(\kappa)),\ldots,\lambda_N(M(\kappa))\}=\sigma(\overline L_{22}).\label{eq:lemma_nm}
    \end{align}
    \end{subequations}
% More precisely,
% \[
% \sigma(\overline L+\kappa \overline P)
% =
% \{\kappa+\lambda_i(\overline L_{11})\}_{i=1}^{m}
% \;\cup\;
% \{\lambda_i(\overline L_{22})\}_{i=1}^{N-m}.
% \]
\end{theorem}
\begin{proof}
The eigenvalues of $\overline M(\kappa)$ solve the equation
\[
\det(\overline M(\kappa)-\lambda I_N)=0,
\]
where 
\[
\overline M(\kappa)=
\begin{pmatrix}
\overline L_{11} +\kappa I_m & \overline L_{12}\\
\overline L_{21} & \overline L_{22}
\end{pmatrix}.
\]
Since we can always find a $\bar \kappa$ such that, for all $\kappa> \bar \kappa$ the block $\overline L_{11} +\kappa I_m$ is invertible, for such values of $\kappa$ we can define 
\[
T(\kappa)=\begin{pmatrix}
I_m & -(\overline L_{11} +(\kappa-\lambda) I_m)^{-1}\overline L_{12}\\
0 & I_{N-m}
\end{pmatrix}. 
\]
Then, we exploit the Schuur complement \cite{schur1917potenzreihen} by considering the following product
\begin{equation}\label{eq:mktk}
\begin{aligned}
&(\overline M(\kappa)-\lambda I_N)T(\kappa)=\\
&\hspace{-4mm}\begin{pmatrix}
\overline M_{11} -\lambda I_m & 0\\
\overline L_{21} & \overline L_{22}-\overline L_{21}(\overline M_{11} -\lambda I_m)^{-1}\overline L_{12}-\lambda I_{N-m} 
\end{pmatrix}
\end{aligned}
\end{equation}
Since
\[
\det((\overline M(\kappa)-\lambda I_N)T(\kappa))=\det((\overline M(\kappa)-\lambda I_N))\det(T(\kappa)),
\]
and observing that $\det(T(\kappa))=1$, 
we obtain that
\begin{equation}\label{eq:factorization}
\det\big((\overline M(\kappa)-\lambda I_N)\big)=\Xi_1 \Xi_2,
\end{equation}
where 
\begin{subequations}
\begin{align}
    \Xi_1&=\det\big((\overline L_{11}+\kappa I_m-\lambda I_m)\big),\\
    \Xi_2&=\det\big(\overline L_{22}-\lambda I_{N-m}-\overline L_{21}(\overline L_{11} +\kappa I_m-\lambda I_m)^{-1}\overline L_{12}\big),\label{eq:xi2}
\end{align}
\end{subequations}
and we used the triangular form of $(\overline M(\kappa)-\lambda I_N)T(\kappa)$, implied trivially from \eqref{eq:mktk}. The eigenvalues of $\overline M(k)$, and then of $M(k)$, are therefore given by the union of the zeros of $\Xi_1$ and $\Xi_2$.

Note that the zeros of $\Xi_1$ are
\[
\kappa+\lambda_i(\overline L_{11}),\quad i=1,\ldots,m,
\]
which implies \eqref{eq:lemma_m}.

Next, take $\lambda=\lambda_i(\overline L_{22})$, for any $i=1,\ldots,N-m$. We then have
\[
\lim_{\kappa\to+\infty}(\overline L_{11} +\kappa I_m-\lambda I_m)^{-1}=0.
\]
From \eqref{eq:xi2}, it follows that, in the limit for $k\rightarrow+\infty$, $\lambda_i(\overline L_{22})$ is a zero of $\Xi_2$, for any $i=1,\ldots,N-m$. This implies \eqref{eq:lemma_nm}, and concludes the proof.
\end{proof}

Now that we have characterized the limit behavior of the spectrum of $M(\kappa)$, we can identify a class of networks whose network pinning controllability can be guaranteed. In particular, we extend the definition of type II master stability function, originally proposed in \cite{boccaletti2006complex}, to the case of network systems on directed hypergraphs
\begin{defi}\label{def:type2}
    We say that the controlled network \eqref{eq:controlled_network} is type II if there exists a real number $\bar \mu$ such that
    \begin{equation}
        \Lambda(\mu)<0, \quad \forall \mu> \bar\mu
    \end{equation}
\end{defi}

Using Theorem \ref{thm:asymptotic_spectrum}, we can now state the following corollary:
\begin{coro}\label{cor:1}
    if network \eqref{eq:controlled_network} is type II and $\Lambda(\lambda_i(\overline L_{22}))\allowbreak<0$ for all $i=1,\ldots,N-m$, then network \eqref{eq:controlled_network} is also pinning controllable.
\end{coro}
\begin{proof}
Since network \eqref{eq:controlled_network} is type II, from Definition \ref{def:type2} it follows that $\lim_{\mu\rightarrow+\infty}\Lambda(\mu)<0$ and $\Lambda(\lambda_i(\overline L_{22}))<0$.
    From Theorem \ref{thm:asymptotic_spectrum}, and since $\Lambda(\lambda_i(\overline L_{22}))<0$ for all $i=1,\ldots,N-m$, Corollary \ref{prop:2} yields the thesis.
\end{proof}
\begin{remark}
Note that, even in the special instance of digraphs, the class of controlled network systems with
a type II master stability function is wider than the class of systems
that can be globally pinning controlled, see e.g. \cite{ancona2024percolation,delellis2018partial},  and references therein. For instance, the conditions in \cite{delellis2018partial} only
include QUAD dynamical systems, whereas type II master stability functions can also be exhibited by non-QUAD dynamical systems, such as the Lorenz system.
\end{remark}

The simplest scenario for a type II network is the classical consensus problem, generalized to the case of a directed many-body topology. Classical leader-follower consensus dynamics, similar to those in \cite{olfati2004consensus}, can be recovered from \eqref{eq:controlled_network} by setting $f=0$ and $g$ as the identity function. We can then obtain the following result:
\begin{coro}\label{thm:consensus}
    If $n=1$, $f=0$, $g$ is the identity, and $\lambda_i(\overline L_{22})>0$ for all $i=1,\ldots,N-m$, then there exists $\kappa$ such that network \eqref{eq:controlled_network} globally achieves consensus onto the pinner's initial condition.
\end{coro}
\begin{proof}
    When $n=1$, $f=0$, the pinner dynamics is constant, that is, $x_p(t)=x_p(0)$ for all $t$. Moreover, as $g$ is the identity, the network dynamics become
    \begin{equation}
        \dot e(t)=-M(\kappa)e(t)
    \end{equation}
From Theorem \ref{thm:asymptotic_spectrum}, and from the continuity of the eigenvalues of $M(\kappa)$ with respect to $\kappa$, the thesis follows. 
\end{proof}

% \color{red}
% PRIMA CONGETTURA DI DE LELLIS: \\
% Una condizione sufficiente sul numero di misure da fare è fare tante misure quanti sono gli autovalori non positivi di $L$.\\
% \color{blue}
% APPENDICE ALLA PRIMA CONGETTURA DI DE LELLIS\\
% Si consideri la rete di sette nodi connessi ad anello nearest neighbour con iperarchi di ordine 3. La Laplaciana del grafo signed associato all'ipergrafo è
% \[
% \begin{bmatrix}
% 1 & -1 & 0.5 & 0 & 0 & 0.5 & -1 \\
% -1 & 1 & -1 & 0.5 & 0 & 0 & 0.5 \\
% 0.5 & -1 & 1 & -1 & 0.5 & 0 & 0 \\
% 0 & 0.5 & -1 & 1 & -1 & 0.5 & 0 \\
% 0 & 0 & 0.5 & -1 & 1 & -1 & 0.5 \\
% 0.5 & 0 & 0 & 0.5 & -1 & 1 & -1 \\
% -1 & 0.5 & 0 & 0 & 0.5 & -1 & 1
% \end{bmatrix}
% \]
% che ha 3 autovalori non positivi, uno 0 e uno con molteplicità due -0.4695.

% Per esaustione si può mostrare che, nonostante siano solo 3 gli autovalori da stabilizzare, non è possibile, pinnando 3 nodi con un arco semplice controllare il sistema (tutti i grounded laplacian togliendo tre righe e tre colonne presentano autovalori non positivi).

% Interessantemente, utilizzando per il pinning una mistura di archi e iperarchi, ad esempio con $\mathcal{P}=\{(1,p),(\{2,3\},p),(\{4,5,6\},p)\}$ il problema di controllo viene risolto. 

% \color{red}

% Una condizione necessaria e sufficiente perché con un dato set di iperarchi/misure riusciamo a controllare la rete è che la ``grounded Laplacian'' (in realtà quella dopo la trasformazione) associata al blocco 0 di $P_{\mathrm{pot}}$ abbia solo autovalori positivi. 
 % \color{black}

\subsection{Select observations}
The application of Corollary \ref{cor:1} yields some important facts about the pinning controllability of network \eqref{eq:controlled_network}, which can be effectively illustrated by means of the following examples. For ease of illustration, we refer to the leader-follower consensus dynamics considered in Theorem \ref{thm:consensus}, so that we can focus our discussion on the spectrum of $M(\kappa)$. In addition, in all examples, we consider homogeneous coupling weights according to Definition \ref{rem:homo}. 

\begin{figure}
\centering
\includegraphics[width=.8\columnwidth]{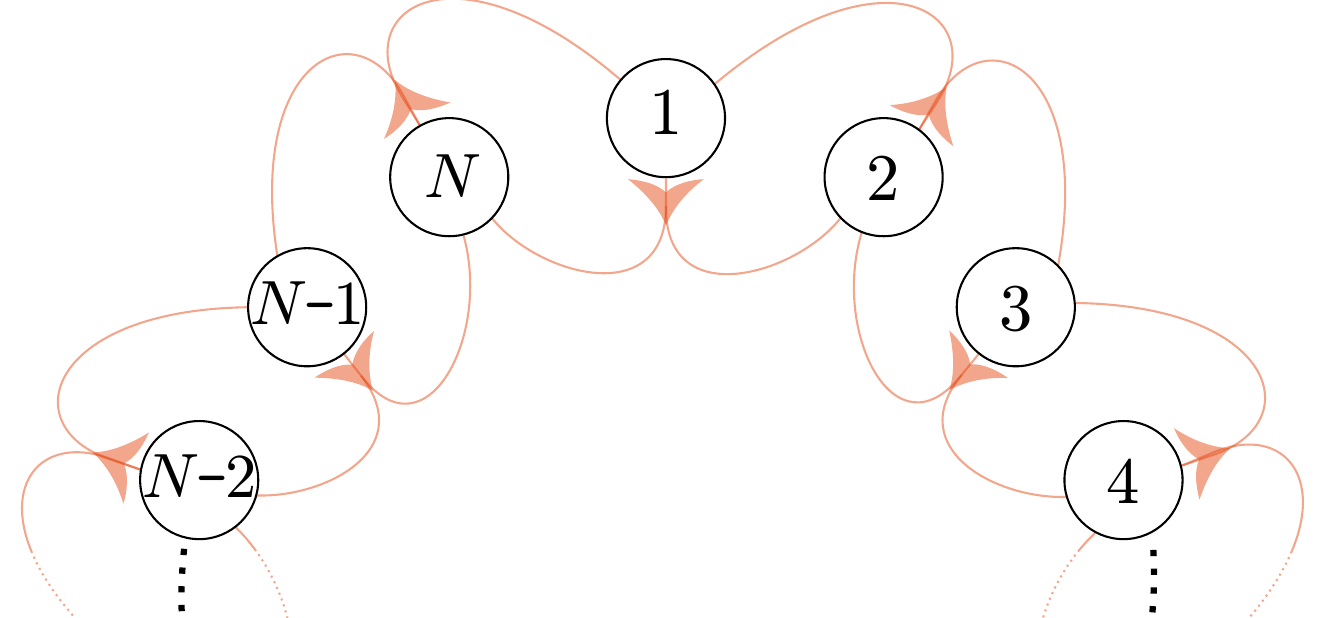}
\caption{Directed three-body nearest neighbor hypergraph.}
\label{fig:nn}
\end{figure}

\textit{Observation 1: using directed hyperedges for pinning may outperform the use of directed edges.} 

\textit{Example 1.} Let us consider an uncontrolled network $\mathscr H_c$ of $N=7$ nodes coupled on a three-body nearest neighbor directed hypergraph (for a description of this topology, see Figure \ref{fig:nn}). Also, let us assume that we can only take three measurements of this network, that is, $|\mathcal E_{\mathrm{pin}}|=3$. Interestingly, we note that if we limit our selection to directed edges as in standard pairwise pinning control (i.e. we individually measure the state of three nodes), there is no selection of three nodes such that $\lambda_i(\overline L_{22})>0$ for all $i=1,\ldots,4$. Numerical explorations confirm that it is impossible to control the network with three standard edges, as illustrated for a representative selection of the pinning nodes in Figure \ref{fig:hyp_better_7}(a). Nonetheless, taking three aggregated measurements we can make all the four eigenvalues of $\overline L_{22}$ positive, thereby allowing the pinner to successfully drive the followers to the reference consensus value, see Figure \ref{fig:hyp_better_7}(b).

\begin{figure}
    \centering
    \includegraphics[width=\linewidth]{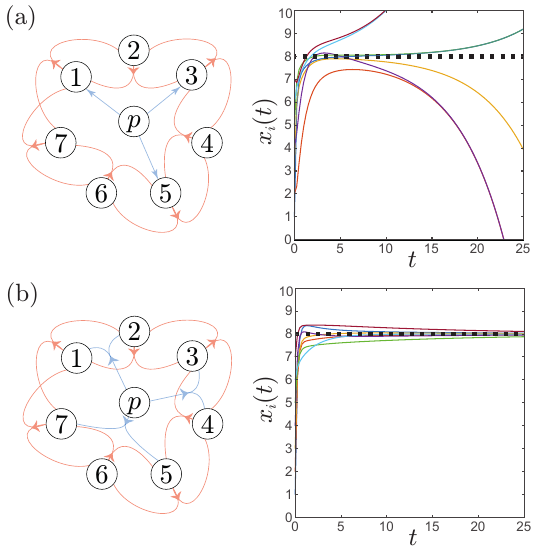}
    \caption{Leader-follower consensus dynamics over a 7-node three-body nearest-neighbor hypergraph. In panel (a), nodes 1, 3, and 5 are controlled through standard, pairwise pinning; in panel (b), each pinning hyperedge corresponds to measuring the average state of two nodes, whereby $\mathcal H(\varepsilon_1^p)=\{1,2\}$, $\mathcal H(\varepsilon_2^p)=\{3,4\}$, $\mathcal H(\varepsilon_3^p)=\{5,7\}$. The left panels depict the topology of the controlled network, whereas the right panels report the state dynamics when $x_i(0)=i,\,i=1,\ldots,7$, $x_{p0}=8$, $\sigma_\varepsilon=1$ for all $\varepsilon\in\mathcal E_c$, and $\kappa=5$, with a thick dashed line identifying the pinner's trajectory. }
    \label{fig:hyp_better_7}
\end{figure}

\textit{Observation 2: in digraphs, leader-follower consensus is attained by pinning a number of nodes equal to the number of non-positive eigenvalues of the Laplacian. The same does not hold for directed hypergraphs.}

\textit{Example 2.} Let us consider the 3-node uncontrolled hypergraph $\mathscr H_c$ in Figure \ref{fig:sipuofare}. Since all the rows sum to zero, and the second and third rows are identical, the Laplacian matrix of the signed graph $\mathscr S(\mathscr H_c)$ associated to $\mathscr H_c$
has two 0 eigenvalues, whereas the third will be positive and equal to 2. Nonetheless, it is possible to control it through only one measurement, that is, $|\mathcal{E}_{\mathrm{pin}}|=1$. For instance, by only controlling node 3, $\overline L_{22}$ has eigenvalues 1 and $0.5$.

\begin{figure}[htb!]
\centering
\begin{tabular}{cc}
\raisebox{-0.65\height}{\includegraphics[width=.2\columnwidth]{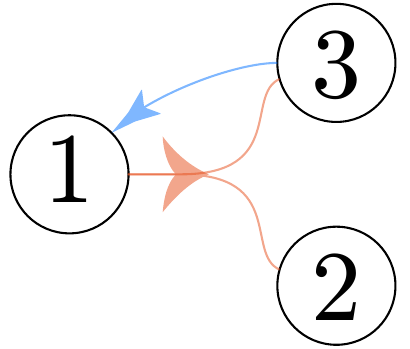}} &
\raisebox{-0.5\height}{
$
L=
\begin{bmatrix}
 1 &  0 & -1   \\
 -1 &  1/2 & 1/2  \\
 -1 &  1/2 & 1/2 
\end{bmatrix}
$
}
\end{tabular}
\caption{Network topology of the controlled network considered in Example 2.}
\label{fig:sipuofare}
\end{figure}

\textit{Example 3.}  Let us consider an uncontrolled network $\mathscr H_c$ of $N=15$ nodes coupled on a three-body nearest neighbor directed hypergraph. The Laplacian matrix of the signed graph $\mathscr S(\mathscr H_c)$ associated to $\mathscr H_c$ has 7 non-positive eigenvalues, but there does not exist any pinning configuration with 7 measurements that can make the eigenvalues of $\overline L_{22}$ all positive, according to our exhaustive numerical search.

\section{Selection of the pinning hyperedges}
%\color{black}
The addition of a new pinning hyperedge assumes the ability of performing an additional measurement, and specifically the measurement of an aggregated state of the heads of the hyperedge. This additional measurement comes at a cost (e.g. the cost of adding an additional sensor), therefore the problem arises of minimizing the number of measurements required to control the network, that is, minimizing the number of pinning hyperedges.

In this vein, we can interpret the set $\mathcal E_{\mathrm{pin}}$ as the set of \textit{potential pinning hyperedges}, that is, the set of possible measurements that can be taken on the hypergraph,
and formulate the following optimization problem:
\begin{equation}\label{eq:optimal}
\begin{aligned}
    &\min \left|\mathcal E_{\mathrm{pin}}^{\mathrm{sub}}\right|\\
    &\text{s.t.}\\
    &\mathcal E_{\mathrm{pin}}^{\mathrm{sub}}\subseteq \mathcal E_{\mathrm{pin}}\\
    &\lim_{\kappa\rightarrow+\infty}\max_{\lambda\in\sigma(M_{\mathrm{sub}})}\Lambda(\lambda)<0,
\end{aligned}
\end{equation}
where $M_{\mathrm{sub}}=L+\kappa P_{\mathrm{sub}}$, with $P_{\mathrm{sub}}$ being the pinning matrix associated to $\mathcal E_{\mathrm{pin}}^{\mathrm{sub}}$.

The optimization problem \eqref{eq:optimal} can in principle be solved by means of Corollary \ref{cor:1}, however it turns out to be combinatorial, whereby all of its possible solutions correspond to $2^m$, where we remind that $m=|\mathcal E_{\mathrm{pin}}|$.

In what follows, we propose a greedy heuristic grounded in Corollary \ref{cor:1} that allows for a computationally efficient selection of the pinning hyperedges.

\subsection*{Heuristic for control}

To heuristically solve the problem, we propose an algorithm that iteratively adds to $\mathcal E_{\mathrm{pin}}^{\mathrm{sub}}$
one of the possible pinning hyperedges in $\mathcal E_{\mathrm{pin}}$. Specifically, at every iteration, we add the pinning hyperedge that maximizes a linear combination between the number of the eigenvalues $\lambda$ of matrix $M$ such that $\Lambda(\lambda)\ge 0$ and the related $\Lambda(\lambda)$ value.
Before using the algorithm, one should preliminarily check the feasibility of the problem by means of Corollary \ref{cor:1}. Specifically, one should verify that selecting $\mathcal E_{\mathrm{pin}}^{\mathrm{sub}}=\mathcal E_{\mathrm{pin}}$, $\Lambda(\lambda_i(\overline L_{22}))<0$ for all $i=1,\ldots,N-m$.

Once this preliminary check is introduced, one could apply the algorithm that we describe in the following:
\begin{enumerate}
\setcounter{enumi}{-1} 
    \item Initialization. 
    %Using Corollary \ref{cor:1} check that, selecting $\mathcal E_{\mathrm{pin}}^{\mathrm{sub}}=\mathcal E_{\mathrm{pin}}$, $\Lambda(\lambda_i(\overline L_{22}))<0$ for all $i=1,\ldots,N-m$. 
    Set $\mathcal E^{\mathrm{sub}}_{\mathrm{pin}}=\emptyset$ and step $k=0$.
    % \item At the $k$-th step, compute $\mathcal U^k=\{\lambda\in\spec(\overline L_{22}^k):\Lambda(\lambda)\ge 0\}$, 
    % where $\overline L_{22}^k$ is the $(N-k)\times(N-k)$ matrix obtained by removing the rows and columns of $\overline L$ corresponding to the  hyperedges in $\mathcal E^{\mathrm{sub}}_{\mathrm{pin}}$, 
    % that is, the set of eigenvalues of $L_{22}^k$ for which the master stability function is positive. We then compute 
    % \[
    % J^k=\sum_{\lambda\in\mathcal U^k}\Lambda(\lambda) + |\mathcal U^k|.
    % \] 
    \item For each $\varepsilon_i\in\mathcal E_{\mathrm{pin}}^{\mathrm{rem}}=\mathcal E_{\mathrm{pin}}\setminus \mathcal E^{\mathrm{sub}}_{\mathrm{pin}}$, compute 
    \[
    J^k_i=\sum_{\lambda\in\mathcal U^k_i} \Lambda(\lambda) + |\mathcal U^k_i|,
    \]
    where $\mathcal U^k_i=\{\lambda\in\spec(\overline L_{22}^{k,i}):\Lambda(\lambda)\ge 0\}$, with $\overline L_{22}^{k,i}$ is the $(N-k-1)\times(N-k-1)$ matrix obtained by removing the rows and columns of $\overline L$ corresponding to the  hyperedges in $\mathcal E^{\mathrm{sub}}_{\mathrm{pin}}\cup \{\varepsilon_i\}$, for $i=1,\ldots,|\mathcal E_{\mathrm{pin}}^{\mathrm{rem}}|$.
    
    \item Compute 
    \[
    i^\star=\argmin_i J^k_i
    \]
    and add $\varepsilon_{i^\star}$ to $\mathcal E_{\mathrm{pin}}^{\mathrm{sub}}$. If $J^k_{i^\star}=0$, terminate the algorithm; otherwise, set $k=k+1$ and go to point 1.
\end{enumerate}

\section{Numerical applications}
Here, we test the performance of the proposed heuristic for control. Specifically, in Section \ref{subsec:lead_foll} we compare it against the pinning strategy proposed in \cite{de2022pinning}, a random sequential selection of the pinning hyperedges in a leader-follower consensus problem, and a selection of the nodes with the largest difference between out- and in-degree. As benchmark hypergraph topology, we focus on directed three-body nearest neighbor hypergraphs, and on random hypergraphs. Then, in Section \ref{subsec:lorenz}, we apply the proposed heuristic to control a network of Lorenz chaotic systems. 
\subsection{Leader-follower consensus}\label{subsec:lead_foll}
\subsubsection*{Directed three-body nearest-neighbor} as a testbed topology for the controlled network, we considered the directed three-body nearest neighbor hypergraph illustrated in Figure \ref{fig:nn}, which is particularly challenging for control. Indeed, even in the case of consensus, that is, when we set $f=0$, $n=1$ and $g$ the identity function in \eqref{eq:uncontrolled_network}, instabilities in the uncontrolled network arise for network sizes larger than 4, as illustrated in the example with $N=6$ in Figure \ref{fig:consensus}.

\begin{figure}
    \centering
    \includegraphics[width=0.9\linewidth]{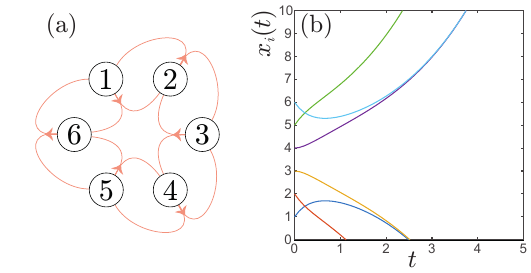}
    \caption{Consensus dynamics in the uncontrolled network \eqref{eq:uncontrolled_network} (with $f=0$, $g(z)=z$, $n=1$). Panel (a) depicts the network topology, which is a 6-node directed three-body nearest neighbor, whereas panel (b) reports the state dynamics when $x_i(0)=i,\,i=1,\ldots,6$ and $\sigma_\varepsilon=1$.}
    \label{fig:consensus}
\end{figure}

Specifically, we consider networks of different size, and vary $N$ between 5 and 20 with step 1. We start by considering the case in which the pinner can potentially measure the state of each node, that is, $\mathcal E_{\mathrm{pin}}=\{\varepsilon_1^p,\varepsilon_2^p,\ldots,\varepsilon_N^p\}$, with $\mathcal H(\varepsilon_i^p)=\{i\}$ for all $i$. For each $N$, we record the minimum number of measurements required to control the network through exhaustive search and compare it with the number obtained by our greedy control heuristic, by the heuristic proposed in \cite{de2022pinning}, and by randomly adding pinning hyperedges, as illustrated in Table \ref{tab:placeholder}. We note that the proposed greedy heuristic strongly outperforms the alternative heuristics. Moreover, it closely matches the performance of an exhaustive search, whereby, for 11 values of $N$ out of 16, the number of selected pinning hyperedges coincides.

To further delve into this finding, for $N=10$ we compare the performance of our greedy heuristic against an exhaustive search considering all possible selections of $\mathcal E_{\mathrm{pin}}$. Namely, each selection corresponds to considering one of the partitions of the node set as the heads of the potential pinning hyperedges, for a total of $115975$ partitions. Out of these partitions, we only consider the $10475$ that fulfill Corollary \ref{cor:1} so that the network can be pinning controlled.
Our control heuristic in $87\%$ of the cases selects the minimum number of pinning hyperedges, and never adds more than one extra hyperedge.

\begin{table}[htbp!]
\definecolor{lightgray}{gray}{0.92}
\centering

\caption{Comparing alternative pinning strategies with respect to the cardinality of $\mathcal E_{\mathrm{pin}}^{\mathrm{sub}}$ on directed nearest-neighbor three-body hypergraphs. From left to right, the first four columns report the number of nodes $N$, the smallest cardinality of $\mathcal E_{\mathrm{pin}}^{\mathrm{sub}}$ obtained through exhaustive search, the one obtained through our proposed heuristic, and according to \cite{de2022pinning}, respectively. The last three columns report the average cardinality of $\mathcal E_{\mathrm{pin}}^{\mathrm{sub}}$ when new pinning hyperedges are added randomly, the probability of obtaining the minimal cardinality, and that of obtaining the second smallest cardinality, respectively.}

\begin{tabular}{ccccr@{\hspace{6pt}}r@{\hspace{6pt}}r}
\textbf{$N$} & \textbf{min} & \textbf{heuristic} & \textbf{\cite{de2022pinning}} & \multicolumn{3}{c}{\textbf{random}} \\
\midrule

5  & 3 & 3 & 3  & $3.55   $   & $\mathit{P}(3)=0.45, $ & $\mathit{P}(4)=0.55$ \\ \rowcolor{lightgray}    
6  & 3 & 3 & 5  & $4.29   $   & $\mathit{P}(3)=0.13, $ & $\mathit{P}(4)=0.45$ \\
7  & 4 & 4 & 6  & $5.13   $   & $\mathit{P}(4)=0.23, $ & $\mathit{P}(5)=0.41$ \\ \rowcolor{lightgray}    
8  & 5 & 5 & 5  & $5.94   $   & $\mathit{P}(5)=0.26, $ & $\mathit{P}(6)=0.54$ \\
9  & 5 & 5 & 7  & $6.81   $   & $\mathit{P}(5)=0.08, $ & $\mathit{P}(6)=0.29$ \\ \rowcolor{lightgray}    
10 & 5 & 6 & 7  & $7.79   $   & $\mathit{P}(5)<0.01, $ & $\mathit{P}(6)=0.13$ \\
11 & 6 & 7 & 8  & $8.48   $   & $\mathit{P}(6)=0.01, $ & $\mathit{P}(7)=0.18$ \\ \rowcolor{lightgray}     
12 & 7 & 7 & 8  & $9.39   $   & $\mathit{P}(7)=0.05, $ & $\mathit{P}(8)=0.14$ \\
13 & 7 & 7 & 11 & $10.26  $   & $\mathit{P}(7)=0.01, $ & $\mathit{P}(8)=0.06$ \\ \rowcolor{lightgray}    
14 & 7 & 7 & 11 & $11.13  $   & $\mathit{P}(7)<0.01, $ & $\mathit{P}(8)=0.02$ \\
15 & 8 & 9 & 12 & $11.75  $   & $\mathit{P}(9)=0.03, $ & $\mathit{P}(10)=0.17$ \\ \rowcolor{lightgray}    
16 & 9 & 9 & 12 & $12.89  $   & $\mathit{P}(9)=0.02, $ & $\mathit{P}(10)=0.02$ \\
17 & 9 & 9 & 14 & $13.95  $   & $\mathit{P}(9)<0.01,$ & $\mathit{P}(10)=0.02$ \\ \rowcolor{lightgray}  
18 & 9 & 9 & 13 & $14.38  $   & $\mathit{P}(9)<0.01,$ & $\mathit{P}(10)=0.01$ \\
19 & 10 & 11 & 14 & $15.91$   & $\mathit{P}(10)<0.01,$ & $\mathit{P}(11)<0.01$ \\ \rowcolor{lightgray}  
20 & 11 & 11 & 14 & $16.32$   & $\mathit{P}(11)<0.01,$ & $\mathit{P}(12)=0.01$ \\
\bottomrule
\end{tabular}
\label{tab:placeholder}
\end{table}

\subsubsection*{ER directed hypergraphs}
Next, we consider the ER directed hypergraphs introduced in \cite{della2023emergence}, and vary the parameter $p$ modulating its density between 0.01 and 0.02, and the maximum order $o$ from 3 to 6 with step 1. For each combination of the pair $(p,o)$, we generate 100 topologies and focus on its giant strongly connected component\footnote{A strongly connected component (SCC) in a hypergraph is defined as a maximal sub-hypergraph such that there is a path from any node to any other node \cite{gallo1993directed}. We call the largest SCC the giant SCC.} and record the average number of pinning edges (which coincide with the cardinality of $\mathcal E_{\mathrm{pin}}^{\mathrm{sub}}$) required to control it, and obtained by i) exhaustive search, ii) our greedy heuristic, iii) the heuristic of \cite{de2022pinning}, iv) a selection of the nodes with the largest difference between out- and in-degree (defined as in Section \ref{subsec:directed_hyper}), and v) a random selection.

The first observation is that, even for 100 node networks, it becomes computationally demanding to determine the nodes to be pinned by exhaustive search (in the sense that for a single network it would scale with $N^{3+|\mathcal E_{\mathrm{pin}}^{\mathrm{sub}}|}$).  Therefore, in our exhaustive search we stopped at exploring pinning selections such that $|\mathcal E_{\mathrm{pin}}^{\mathrm{sub}}|\le 4$. As a consequence, for $p=0.01$ and $o\leq 5$, we only obtained a lower bound for the smallest percentage of nodes that need to be pinned, as highlighted in Table \ref{tab:placeholder2}. Nonetheless, our greedy heuristic closely matches the performance of an exhaustive search. Also, it clearly outperforms the alternative heuristic proposed in \cite{de2022pinning}, and a purely topological selection criterion based on the difference between out- and in-degree. Predictably, a random selection yields the worst performance.

% \textcolor{red}{Magari, a rinforzare il fatto che le NN sono particolarmente difficili, si può far notare che nel caso peggiore per le ER abbiamo bisogno del 7\% della rete, mentre per le NN stiamo sempre sopra al 50\%...}

% \textcolor{red}{Qui secondo me si potrebbe anche sottolineare che se $o=2$ tutti gli algoritmi ci davano come risultato 1, per rinforzare come il problema sia fondamentale per gli ipergrafi (mentre risolto per i grafi). Inoltre, ci fa da ponte con la prossima sezione, in cui in realtà non esiste una soluzione nemmeno per i grafi...}

% \begin{table}[htbp]
% \definecolor{lightgray}{gray}{0.92}
% \centering
% \begin{tabular}{cccccccc}
% \textbf{$p$} & \textbf{$o$} & $\langle N \rangle$ &\textbf{min} & \textbf{heuristic} & \textbf{\cite{de2022pinning}} & \textbf{$d^{\mathrm{out}}-d^{\mathrm{in}}$} &
% \textbf{random} \\
% \midrule
% 0.01 & 3 & 75.4 & 5.4 (20) & 5.5 & 9.9 & 35.3 & 54.9    \\ \rowcolor{lightgray}    
% 0.02 & 3 & 97.5 & 1.1 (100) & 1.1 & 1.6 & 5.8 & 10.4  \\
% 0.01 & 4 & 93.4 & 2.9 (50) & 3.4 & 6.3 & 21.9 & 46.5  \\ \rowcolor{lightgray}    
% 0.02 & 4 & 99.4 & 1.0 (100) & 1.0 & 1.2 & 1.2 & 1.6  \\
% 0.01 & 5 & 97.7 & 1.6 (87) & 1.7 & 3.3 & 13.3 & 27.0  \\ \rowcolor{lightgray}    
% 0.02 & 5 & 100.0 & 1.0 (100) & 1.0 & 1.4 & 1.0 & 1.3  \\
% 0.01 & 6 & 99.3 & 1.1 (100) & 1.1 & 2.1 & 3.1 & 9.6  \\ \rowcolor{lightgray}
% 0.02 & 6 & 100.0 & 1.0 (100) & 1.0 & 1.3 & 1.0 & 1.3  \\ 
% \bottomrule
% \end{tabular}
% \caption{qui riporto i numeri}
% \label{tab:placeholder}
% \end{table}

\begin{table}[htbp!]
\definecolor{lightgray}{gray}{0.92}
\centering

\caption{Comparing alternative pinning strategies with respect to the cardinality of $\mathcal E_{\mathrm{pin}}^{\mathrm{sub}}$ on random directed ER hypergraphs. From left to right, the first three columns report the parameter $p$ of the ER hypergraph, the maximum order $o$ of the hypergraph, the average size (over 100 repetitions) of the giant strongly connected component, respectively. The last five columns report the smallest percentage of nodes in $\mathcal E_{\mathrm{pin}}^{\mathrm{sub}}$ to pinning control the network obtained by exhaustive search (or a lower bound when computation is unfeasible), when new pinning hyperedges are added using the proposed greedy heuristic, according to \cite{de2022pinning}, maximizing $\Delta d= d^{\mathrm{out}}-d^{\mathrm{in}}$, or randomly, respectively.}

\begin{tabular}{cccccccc}
\textbf{$p$} & \textbf{$o$} & $\langle N \rangle$ &\textbf{min} & \textbf{heuristic} & \textbf{\cite{de2022pinning}} & \textbf{$\Delta d$} &
\textbf{random} \\
\midrule
0.01 & 3 & 75.5 & $\ge$5.5 (42) & 7.3 & 12.6 & 46.2 & 74.6    \\ \rowcolor{lightgray}    
0.02 & 3 & 97.5 & 1.2 (100) & 1.2 & 1.7 & 6.0 & 10.7 \\
0.01 & 4 & 93.1 & $\ge$2.7 (90) & 3.2 & 7.1 & 27.9 & 50.0   \\ \rowcolor{lightgray}    
0.02 & 4 & 99.4 & 1.0 (100) & 1.0 & 1.2 & 1.2 & 1.6  \\
0.01 & 5 & 97.5 & $\ge$1.6 (98) & 1.7 & 4.0 & 13.4 & 26.7  \\ 
\rowcolor{lightgray}    
0.02 & 5 & 100.0 & 1.0 (100) & 1.0 & 1.4 & 1.0 & 1.3  \\
0.01 & 6 & 99.3 & 1.1 (100) & 1.1 & 2.1 & 3.1 & 9.7  \\ \rowcolor{lightgray}
0.02 & 6 & 100.0 & 1.0 (100) & 1.0 & 1.3 & 1.0 & 1.3  \\ 
\bottomrule
\end{tabular}

\label{tab:placeholder2}
\end{table}

\subsection{Pinning synchronization of Lorenz systems}\label{subsec:lorenz}
Here, we show how the proposed heuristic can be used to select the nodes to be pinned in a network of nonlinear dynamical systems.
\subsubsection{Individual dynamics}
the nodes are Lorenz systems \cite{lorenz1963deterministic}, that is, in \eqref{eq:controlled_network} we set the vector field $f$ to
\begin{equation}
\label{eq:Lorenz}    
f(x_i)=\left[
\begin{array}{c}
\mathfrak{s} (x_{i2}-x_{i1}) \\[-1mm]
\mathfrak{s} x_{i1}-x_{i2}-x_{i1} x_{i3}\\ [-1mm]
x_{i1} x_{i2}-\mathfrak{b}(x_{i3} + \mathfrak{p}+\mathfrak{s})
\end{array}\right],
\end{equation}
where the parameters $\mathfrak{b}=8/3$, $\mathfrak{p}=28$, and $\mathfrak{s}=10$ are selected so that, in the absence of coupling and control, the node dynamics admits a chaotic attractor. 

\subsubsection{Coupling protocol}
we consider that the nodes are coupled through the interaction function $g(x)=\arctan(x)$. As illustrated in panels (a) and (b) in Figure \ref{fig:lorenz_control}, the controlled network \eqref{eq:controlled_network} is type II according to Definition \ref{def:type2}.
\subsubsection{Hypergraph topology and coupling gain}
we consider a strongly connected directed ER hypergraph as introduced in \cite{della2023emergence}, with $N=100$ nodes, parameter $p=0.01$, and maximum order of interaction $o=4$. Selecting as coupling gain $\sigma_\varepsilon=30$ for all $\varepsilon\in \mathcal E_c$, we note that there are 10 eigenvalues $\lambda\in\sigma(L)$ such that $\Lambda(\lambda)>0$, see Figure \ref{fig:lorenz_control}(a), where each white cross corresponds to at least one eigenvalue, with some of the eigenvalues being repeated.  

%in modo che l'interazione tra i nodi sia abbastanza forte per vincere la tendenza dei sistemi caotici di divergere. Alcuni autovalori della rete (quelli con parte reale minore di 40) sono mostrati nel pannello \ref{fig:lorenz_control}a della figura con delle x bianche: 5 autovalori sono minori o uguali a zero, e 10 autovalori hanno associato una $\Lambda>0$. 

\subsubsection{Selection of the pinning hyperedges} by applying our heuristic, we identify 14 nodes that need to be pinned to guarantee that $\Lambda_{\max}<0$ for a sufficiently large control gain. Figure \ref{fig:lorenz_control}(b) illustrates how the eigenvalues of $P$ are shifted so that $\Lambda_{\max}=-0.08<0$. Using this selection of pinned nodes, we then run a simulation in which we take initial conditions randomly extracted from a Gaussian with zero mean and variance 100, and set the control gain as $\kappa=60$.We observe how all the state trajectories converge to each other as the error norm approaches zero, as illustrated in Figure \ref{fig:lorenz_control}(c)-(d).

% Applichaimo la nostra euristica. La rete si riesce a controllare pinnando 14 nodi, dove andranno a convergere gli autovalori quando il guadagno va all'infinito (sempre quelli con Re<40) è mostrato \ref{fig:lorenz_control}b. Il massimo valore di $\Lambda$ valutato tra gli autovalori di $L_{22}$ è -0.008. 
% Nei pannelli \ref{fig:lorenz_control}c e \ref{fig:lorenz_control}d riportiamo una simulazione del sistema fatta partendo da condizioni iniziali casuali (estratte da una normale di media nulla e varianza 100), e selezionando $\kappa=60$. Nel pannello c riportiamo l'andamento della variabile 1 dei 100 sistemi, mentre nel pannello d riportiamo (in scala logaritmica) la norma dell'errore di pinning per ciascuno dei nodi.  

\begin{figure}[tbp!]
    \centering
    \includegraphics[width=\linewidth]{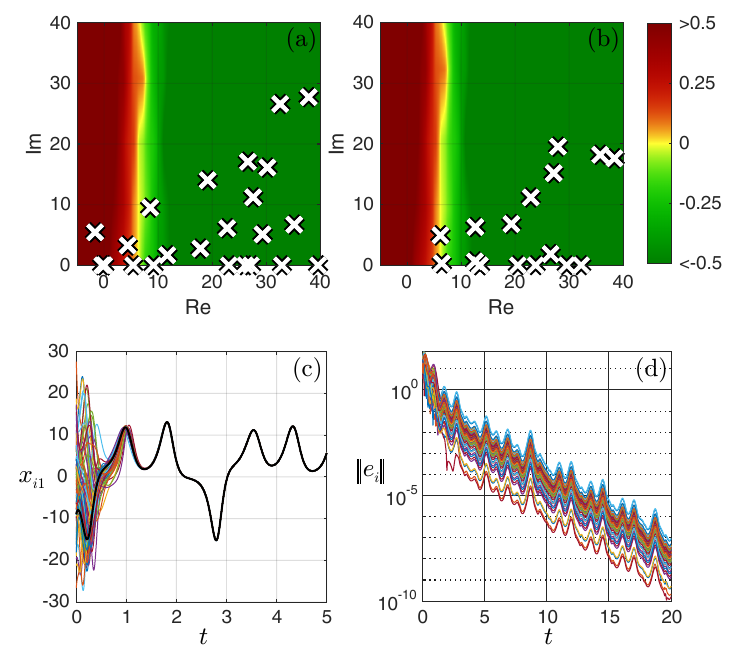}
    \caption{Networks of $N=100$ Lorenz systems coupled through a directed ER hypergraph with parameters $p=0.01$ and $o=4$. Panels (a) and (b) superimpose the eigenvalues of $L$ and $M$ to the Master Stability Function obtained for $g(z)=\arctan z$, respectively. Each eigenvalue is associated to a white cross, some of them are repeated, and some other (associated with negative values of $\Lambda$) are not in the plot as they have real or imaginary part larger than $40$. Panels (c) and (d) depict the time evolution of the first state variable of each system and of the node error norm, respectively.}
    \label{fig:lorenz_control}
\end{figure}

\section{Conclusions}

%\color{black}
In this paper, we focused on the pinning control problem for networks of systems coupled through higher-order interactions that can be modeled through directed hypergraphs. In particular, we have considered the optimal control problem of minimizing the number of measurements required to steer the network nodes towards the desired trajectory set by the pinner. We have modeled each measurement as a pinning hyperedge, and then translated the problem into the minimization of the number of pinning hyperedges required to control the network, a problem that had no solution in the literature.

To solve this problem, we have studied the limit behavior of the eigenvalues of an extended Laplacian-like matrix, which captures the topology of the controlled hypergraph together with the select pinning configuration. For the wide class of networks having a type II Master Stability Function, this allowed us to find the minimal set of measurements (pinning hyperedges) required to control the entire network. We have shown that the use of pinning hyperedges instead of standard edges can reduce the number of measurements required.

Since the optimal solution is computationally daunting for large networks, we have proposed an efficient heuristic to find a mathematically grounded suboptimal solution. We have shown how this solution closely matches the optimal one on testbed directed hypergraphs. Moreover, when used for the selection of standard pinning edges (instead of hyperedges), it strongly outperforms the only pre-existing heuristic \cite{de2022pinning}, as well as alternative, purely topological strategies.

%INIZIO RAGIONAMENTO FRANCESCO

%Se capisco bene, noi vorremmo risolvere in $\tilde v$ l'equazione
%\begin{equation}\label{eq:principale}
%    (P-\labmda_1I)\tilde v = (L-\tilde \lambda I)v_0
%\end{equation}
%sottoposto a
%\begin{equation}
%    (P-\lambda_1I)v_0 = 
%\end{equation}
%ossia che $v_0$ è un autovettore di $P$ associato a $\lambda_1$. Ora, consideriamo il caso di grafi normali, e assumiamo che $\lambda_1=1$. In questo caso sappiamo benissimo che $v_0$ appartiene allo span dei nodi pinnati e quindi all'autospazio associato all'autovalore nullo di $L$. Ponendo quindi $\tilde \lambda = 0$ direi che l'equazione è sempre verificata per $\tilde v = 0$. Ora consideriamo il caso $\lambda_1=0$. In questo caso $v_0$ appartiene allo span dei nodi non pinnati, e genericamente oserei dire che non è facile sapere a che sottospazio appartenga $(L-\tilde \lambda I)v_0$, salvo nel caso in cui $\tilde \lambda = 0$, nel qual caso sicuramente $(L-\tilde \lambda I)v_0$ rimane nello span dei nodi non pinnati. Tuttavia, in questo caso mi sembra che $(P-\labmda_1I)\tilde v$ appartenga necessariamente al sottospazio ortogonale a quello dei nodi non pinnati in quanto questo sottospazio mi pare proprio il nummo di $(P-\labmda_1I)$.

\bibliographystyle{IEEEtran}

\end{document}